\documentclass[reqno,a4paper]{amsart}
\usepackage{amssymb,amsmath,amsfonts,amsthm}
\usepackage[margin=1in]{geometry}
\usepackage{hyperref}
\numberwithin{equation}{section}
\newtheorem{Thm}{Theorem}[section]

\newtheorem{Lemma}{Lemma}[section]

\usepackage[abbrev,citation-order]{amsrefs}

\usepackage{authblk}
\usepackage[foot]{amsaddr}

\title{Static Stellar Phase Transitions in General Relativity and a Generalized Buchdahl Limit}

\author[M.\ Reintjes]{Moritz Reintjes$^*$}
\address[*]{Department of Mathematics\\ City University of Hong Kong \\ Hong Kong}
\email{moritzreintjes@gmail.com}

\author[R.\ Xia]{Ruochen Xia$^*$}
\thanks{Both authors contributed to all aspects of the paper. }
\date{November 18, 2025}

\begin{document}

\maketitle

\begin{abstract}
We give the first general construction of solutions of the static spherically symmetric Einstein-Euler equations, the Tolman-Oppenheimer-Volkoff (TOV-)equation, with prescribed density functions allowed to be discontinuous and non-uniform; these solutions describe stellar phase transitions in General Relativity. Boundedness of the resulting pressure functions solving the TOV-equations, from the boundary down to the stellar center, is obtained by identifying a novel condition on the prescribed density, in generalization of the classical Buchdahl limit. Moreover, we introduce a new necessary condition for the existence of such bounded pressure functions, which in the special case of a uniform density state reduces to the classical Buchdahl limit on the stellar mass-radius relationship. We present various examples to study the stellar mass-radius relationships resulting from our new conditions.
\end{abstract}


\section{Introduction} \label{Sec_Intro}

The study of relativistic stars goes back to K. Schwarzschild's 1916 paper \cite{schwarzschild1916}, where he derived an explicit solution of the static spherically symmetric Einstein-Euler equations (i.e., with energy-momentum tensor of a perfect fluid) in terms of the fluid pressure, assuming a prescribed constant density function. In Schwarzschild's stellar model, the pressure vanishes at the surface of the star, where the interior spacetime metric matches the Schwarzschild metric continuously, but the constant density jumps {\it discontinuously} to zero. This is the first fluid dynamical phase transition constructed in General Relativity (GR), describing the transition from a uniform (constant) density state to vacuum \cite{Choquet, Weinberg}. 

The general theory of static spherically symmetric solutions of the Einstein-Euler equations, describing the interior of static fluid spheres without viscosity, was developed in \cite{Tolman, OV}. This led to the formulation of the famous Tolman-Oppenheimer-Volkoff (TOV-)equation in 1931, a non-linear first order ordinary differential equation (ODE), relating the pressure $p$ with the density $\rho$ of a fluid co-moving with spacetime. It is well-known that it is in general not feasible to find explicit solutions in closed form of the TOV-equation due to its non-linear structure \cite{martins2019}, and significant efforts have been taken to find algorithms for generating new solutions from known ones \cite{generating-solutions2002, generating-solutions2003, generating-solutions2004, generating-solutions2005, generating-solutions2007, generating-solutions2008}. 

The first rigorous existence theory for the TOV-equation was only given in 1991 by Rendall and Schmidt \cite{Rendall} for a prescribed barotropic equation of state, ($p=p(\rho)$ or equivalently $\rho=\rho(p)$), which is assumed to be monotonous and continuous. This assumption of continuous monotonicity rules out phase transitions. For example, phase transitions require a change in monotonicity of $p(\rho)$, like the up-down-up shape of the van-der-Waals equation of state, in order to connect discontinuities in the density as constant pressure states.  The rigorous construction of stellar phase transitions in GR was so far achieved in special settings, often to model neutron stars, in particular, Christodoulou's two phase model \cite{christodoulou_self-gravitating_1996}, as well as constructions for prescribed special density functions (with the pressure as the unknown in the TOV-equation), matching constant density states across a single discontinuity (the phase transition) \cite{seidov1971, Lindblom, zdunik1987}, or across a finite number of discontinuities \cite{mehra_1969}, or matching a constant density state to density functions with quadratic decay \cite{durgapal_1969, durgapal_1971, gehlot_1971, mehra_1981}. Despite the significant recent progress made in the study of the TOV-equations (see \cite{smoller1994, smoller1997, smoller1998, generating-solutions2002, generating-solutions2003, generating-solutions2004, generating-solutions2005, generating-solutions2007, generating-solutions2008, Martin_2003, andersson_asymptotic_2019, martins2019, harko_exact_2016, heinzle_finiteness_2002} and references therein), the question under what conditions the TOV-equations with general prescribed {\it non-uniform}  and {\it discontinuous} density functions---describing general phase transitions---are solvable, and how to construct their solutions, remains open. 

In this paper we resolve this question by introducing a checkable condition, sufficient for the existence of solutions with finite pressure $p(r)$, which generalizes the classical Buchdahl limit \cite{Buchdahl} to general non-uniform and discontinuous density functions.   Based on this condition, we then give the first construction of general static stellar phase transitions in GR, by constructing solutions $p(r)$ of the TOV-equations with prescribed density, allowing for general discontinuous and variable density functions, and by proving global existence of solutions all the way from the vacuum boundary down to the stellar center at $r=0$. Moreover, we introduce a new necessary condition on a prescribed density for the existence of such bounded pressure functions, which in the special case of a constant density state reduces to the classical Buchdahl limit on the stellar mass-radius relationship. 

Let us remark that it is well-known that solutions of the TOV-equation (a differential equation of Ricatti form) may blow-up. Thus identifying conditions to avoid blow-up of solutions is central to any general construction of solutions of the TOV-equation. This is accomplished in this paper for general prescribed non-uniform density functions containing a finite number of jump discontinuities, and our condition guarantees finiteness of the pressure. Analogously, in \cite{Rendall, heinzle_finiteness_2002, andersson_asymptotic_2019} conditions ware identified to guarantee solutions of finite extend.            

In Section \ref{Sec_Prel}, we review the TOV-equations and introduce stellar phase transitions based on Israel's matching condition \cite{israel1966}. In Section \ref{Sec_Results}, we state our results, including our sufficient condition (Theorem \ref{Thm_existence}) as well as our necessary condition (Theorem \ref{Thm_necessary}). In Section \ref{Sec_staircase-density}, we solve the TOV-equation for the special case of piece-wise constant density functions explicitly, based on introducing a variable change as well as a matching construction, (see \cite{mehra_1969} for an analogous construction). Based on this, we prove in Section \ref{Sec_proof_suff} existence of solutions $p(r)$ to the TOV-equations for density functions satisfying our sufficient condition \eqref{Thm1}, describing general static phase transition in GR, thereby establishing our main result, Theorem \ref{Thm_existence}. This proof is constructive, based on establishing convergence of the solutions constructed in Section \ref{Sec_staircase-density};  authors expect this can be implemented numerically. In Section \ref{Sec_proof_nec}, we give the proof of our necessary condition, Theorem \ref{Thm_necessary}. In Section \ref{Sec_app_examples}, we work out various examples of density profiles of stellar structures to illustrate the relation between our necessary and sufficient conditions and the classical Buchdahl limit, as well as their implications to mass-radius relationships.

\section{Preliminaries} \label{Sec_Prel}

In this paper we address the coupled Einstein-Euler equations
\begin{eqnarray} \label{Einstein-Euler}
	G_{\mu\nu} = 8\pi T_{\mu\nu}   
	\hspace{1cm} \text{and} \hspace{1cm}
	{T^{\mu\nu}}_{; \nu} = 0,
\end{eqnarray}
where $G_{\mu\nu}\equiv R_{\mu\nu} - \frac12 R g_{\mu\nu}$ is the Einstein tensor and $T^{\mu\nu} \equiv (\rho+p)u^{\mu}u^{\nu}+pg^{\mu\nu}$ is the energy-momentum tensor of a perfect fluid with pressure $p$, density $\rho$ and fluid $4$-velocity $u^\mu$ normalized to $u^\sigma u_\sigma =-1$. We adopt geometric units, i.e., the speed of light is $c=1$ and Newton's constant is $\mathcal{G}=1$. We further use the Einstein convention of summing over repeated upper and lower indices. We assume the metric tensor $g = g_{\mu\nu}dx^\mu dx^\nu$ is static and spherically symmetric, and make the ansatz
\begin{equation}\label{TOV_metric}
	g = - e^{2\nu(r)} dt^2 + e^{2\lambda(r)} dr^2 + r^2 (d\theta^2 + \sin^2(\theta) d\phi^2).
\end{equation}
We assume further that the fluid $4$-velocity is co-moving, i.e., $u^0 = e^{-\nu}$ and $u^j=0$ for $j=1,2,3$. It was shown by Tolman, Oppenheimer and Volkoff in \cite{Tolman} that the Einstein-Euler equations are then equivalent to
\begin{eqnarray}   
	\frac{d m}{dr} &=& 4\pi r^2 \rho \label{TOV2} \\
	\frac{d \nu}{dr}  &=& - \frac{dp}{dr} (p+\rho)^{-1}  \label{TOV1}    \\
	\frac{dp}{dr} &=& - \frac{\rho m}{r^2} \Big( 1 + \frac{p}{\rho} \Big)  \Big(1+ \frac{4\pi r^3 p}{m} \Big)  \Big( 1 - \frac{2m}{r} \Big)^{-1}, \label{TOV3}
\end{eqnarray}
assuming the mass function $m(r)$ and metric component $e^{2\lambda}$ are related by $e^{2\lambda} = \left(1 - \frac{2m}{r}\right)^{-1}$, $r>2m(r)$. Equation \eqref{TOV3} is the celebrated {\it TOV-equation}, a non-linear first order system of ODE's.

There are two common approaches for studying equations \eqref{TOV1} - \eqref{TOV3}. The first approach prescribes a barotropic equation of state $p=p(\rho)$, (or $\rho=\rho(p)$), and seeks to solve the TOV-equation \eqref{TOV3} for $\rho$ (or $p$) as the unknown, such that $\rho$ (or $p$) is bounded and has finite extend (i.e., a finite stellar radius). The second approach prescribes the density function over a finite domain (the radius of the star), and seeks to solve the TOV-equation \eqref{TOV3} for the pressure function (instead of prescribing it via an equation of state), such that the pressure function stays finite over the domain of the star as a condition of physical admissibility of the density profile.\footnote{Note that, for a prescribed positive density function, \eqref{TOV3} implies that $p(r)$ is strictly monotonously decreasing. Thus, if $\rho$ is assumed to be strictly monotonously decreasing, then one can solve for the radius as $r=r(\rho)$ and recover a barotropic equation of state $p=p(\rho)$, or alternatively in the form $\rho=\rho(p)$.}

Throughout this paper we adopt the latter approach, that the density function is prescribed, and that the pressure function is the unknown of the TOV-equation. Then, equation \eqref{TOV2} can be integrated directly as 
\begin{equation}\label{m(r)}
	m(r) = \int^r_0 4\pi r^2 \rho(r) dr,
\end{equation}
and equation \eqref{TOV1} can be solved for $\nu(r)$ once the TOV-equations \eqref{TOV1} is solved in terms of the pressure $p(r)$. For a prescribed density, the TOV-equation turns into an ODE of Ricatti-type, a non-linear equation which is known for blow-up of solutions, and thus the main difficulty lies in establishing global existence of solutions, all the way from the vacuum boundary to the center of the star.

The objective of this paper to give a general construction of solutions to \eqref{TOV1} - \eqref{TOV3}, starting with a given bounded density function. By allowing for discontinuities in the density, we also include fluid dynamical phase transitions. That is, discontinuities in $\rho$ give rise to Lipschitz continuous metrics \cite{Reintjes:2021rcz, Reintjes:2022rmh, ReintjesTemple_essreg}, and by Israel's matching theory \cite{israel1966} we obtain a weak solution of the Einstein-Euler equation as long that smooth (classical) solutions are matched across the surfaces where $\rho$ is discontinuous such that the Israel junction conditions (also known as the Rankine-Hugoniot conditions) hold,
\begin{equation}\label{RH-condition}
	[T^{\mu\nu}] N_\nu =0,
\end{equation} 
where $N^\mu$ is the unit normal vector to the surface of discontinuity. In the case of static spherically symmetric metrics \eqref{TOV_metric} and co-moving fluids, \eqref{RH-condition} reduces to the condition that the pressure be continuous across discontinuities in the density,
\begin{equation}\label{RH_p_cont}
	[p]=0.
\end{equation} 
When $p=p(\rho)$, condition \eqref{RH_p_cont} can only be met across a discontinuity in the density if $p=p(\rho)$ changes its monotonicity in between the two density states, which constitutes a fluid dynamical phase transition. For example, by the up-down-up shape of the van-der-Waals equation of state $p(\rho) = \frac{\rho}{\beta - \rho} - \alpha \rho^2$, one can match $p$ continuously across any discontinuity in the density as long that the two density states lie on different sides of the critical region where $p'(\rho)<0$. We conclude that any continuous solution $p$ of the TOV-equation for a given discontinuous density function $\rho(r)$ describes a phase transition at each discontinuity in $\rho$. Note, the entropy equation $(s u^\mu)_{;\mu}=0$ automatically implies that $[s]=0$ holds as an identity across phase transitions in our static spherically symmetric setting. Our strategy in this paper is to start with a given density function $\rho(r)$, containing multiple discontinuities, such that $\rho$ vanishes beyond some finite $R>0$ (the radius of the star). We then prove that, subject to the integral condition \eqref{Thm1} below, one can always solve the TOV-equation \eqref{TOV3} with a Lipschitz continuous pressure function $p(r)$. This describes a stellar phase transition.

\section{Statement of Results} \label{Sec_Results}

We now introduce and state our main results, Theorems \ref{Thm_existence} and \ref{Thm_necessary} below. For this, we consider a given density function $\rho(r)$, defined for $r\in [0,R]$, such that $\rho(r)=0$ for all $r>R$; that is, $R>0$ is the radius of the star and $\rho(R)$ may be non-zero. We assume throughout the paper that $\rho(r)$ is bounded, monotonously decreasing (i.e., non-increasing), and piece-wise continuous, containing at most a finite number of discontinuities, (representing a finite number of phase transitions). Moreover, we assume throughout that the density function is such that the mass distribution lies everywhere outside of its Schwarzschild radius, in the sense that there exists some constant $K \in (0,1)$, such that\footnote{Note that a violation of \eqref{Outside_Schwarzschild} appears to lead to singular, potentially un-physical, stellar configurations. Note further that \eqref{Outside_Schwarzschild} still holds at $r=0$ by boundedness of $\rho(r)$, since $\lim_{r\to0}\frac{2m(r)}{r}=\lim_{r\to0} 8\pi r^2 \rho(r) =0$.}
\begin{equation} \label{Outside_Schwarzschild}
	1 - \frac{2m(r)}{r} > K^2.
\end{equation}
Now, in order to satisfy the Rankine-Hugoniot jump condition \eqref{RH_p_cont} at the surface of the star, we need to choose for $p$ the boundary data $p(R)=0$. Likewise, for the geometry outside of the star to agree with the Schwarzschild metric, we impose the data $e^{2\nu(R)} = \left(1-\frac{2m(R)}{R}\right)$. Our existence result for the TOV-equation is the following:

\begin{Thm}  \label{Thm_existence}
Assume $\rho(r)$ is piece-wise continuous, non-increasing and bounded on $[0,R]$, $\rho(r)>0$ for all $r\in [0,R)$ and $\rho(r)=0$ for all $r>R$, and assume \eqref{Outside_Schwarzschild} holds. Assume further there exists a constant $\Delta\in(0,1)$ such that 
	\begin{equation}\label{Thm1}
		\int^{R}_0 \frac{m(r)}{r^2} \left(1-\frac{2m(r)}{r}\right)^{-\frac{3}{2}} dr \ \leq \ \Delta ,
	\end{equation}
where the mass function $m(r)$ is given by \eqref{m(r)}. Then there exists a unique solution $\{p(r),\nu(r)\} \in C^{0,1}([0,R])$ of the static spherically symmetric Einstein-Euler equations, \eqref{TOV1} - \eqref{TOV3}, satisfying the boundary data $p(R)=0$ and $e^{2\nu(R)} = \left(1-\frac{2m(R)}{R}\right)$. Moreover, the pressure function $p(r)$ is everywhere non-negative, strictly monotonously decreasing and bounded on $[0,R]$ with bound 
\begin{equation} \label{bound_p_Thm1}
p(r) \ \leq \ \frac{\Delta }{e^{\lambda(r)}-\Delta}\rho(r) \ \leq \ \frac{\Delta}{1-\Delta} \rho(r), \ \ \ \ \text{for all} \ r\in [0,R].   
\end{equation}
\end{Thm}

The solutions constructed in Theorem \ref{Thm_existence} describe gaseous phase transitions, since $p$ is continuous across each discontinuity in $\rho$, cf. Section \ref{Sec_Prel}. The resulting metric tensor \eqref{TOV_metric} is Lipschitz continuous and agrees with the Schwarzschild metric for any $r \in [R,\infty)$. Note that \eqref{bound_p_Thm1} implies that the central pressure $p_c$ is bounded by the central density $\rho_c$ in terms of $p_c \leq \frac{\Delta }{1-\Delta}\rho_c$.
	
Condition \eqref{Thm1} is the {\it sufficient} condition we identify for existence of solutions with bounded pressure function. It generalizes the classical Buchdahl condition \cite{Buchdahl} from a given single constant density to general varying (non-uniform) density functions. In more detail, in the case of a single {\it constant} density, Buchdahl showed that 
\begin{equation} \label{Buchdahl}
	\frac{2M}{R} < \frac{8}{9} 
\end{equation}
is a necessary and sufficient condition for the central pressure $p(0)$ of a star of total mass $M$ and radius $R$ to be finite. This implies that stars of constant density with a radius smaller than $9/8^{th}$ of its Schwarzschild radius $2M$ cannot exist, since no finite central pressure would suffice to counter the gravitational force to uphold the star. The following example illustrates that our condition is a refinement of the classical Buchdahl limit as a sufficient condition. For this, consider density functions with the property that its mass function is of the form
\begin{equation}
	m(r)=Ar^\alpha,
\end{equation} 
for constants $\alpha>1$ and $A>0$. Then, for $1<\alpha\leq3$, taking here for ease of presentation $\Delta=1$, condition \eqref{Thm1} takes on the form 
\begin{equation} \label{Buchdahl_refinement}
	\frac{2M}{R} < 1-\frac{1}{\alpha^2},
\end{equation} 
where $M\equiv m(R)$, (see Section \ref{Sec_app_examples} for details). When $\alpha=3$, the density is a constant and \eqref{Buchdahl_refinement} reduces to the classical Buchdahl condition \eqref{Buchdahl}, while for $1<\alpha<3$, \eqref{Buchdahl_refinement} gives a refinement of the classical Buchdahl limit as a sufficient condition.\footnote{Note that the cases $\alpha>3$ and $\alpha<1$ appear unphysical, since for $\alpha>3$ the density profile $\rho(r)$ is increasing, while for $\alpha<1$ the metric component $e^{2\lambda(r)}$ vanishes as $r\to 0$.} 

Our second main result gives, in equation \eqref{necessary_cond_Thm} below, a generalization of the Buchdahl limit as a {\it necessary} condition for boundedness of the pressure:

\begin{Thm} \label{Thm_necessary}
Assume $\rho(r)$ is a bounded, non-increasing, non-negative, piece-wise differentiable density function with at most finitely many jump discontinuities on $[0,R]$, such that $\rho(r)=0$ for all $r>R$, and such that \eqref{Outside_Schwarzschild} holds for all $r\in[0,R]$. Assume $p(r)$ is a non-negative Lipschitz continuous solution of the TOV equation \eqref{TOV3} with $p(R)=0$ and which is bounded on $[0,R]$. Then the density function $\rho(r)$ satisfies 
\begin{equation}\label{necessary_cond_Thm}
	\rho(r) > e^{-\lambda(r)}\int^{R}_r \frac{m(r)\rho(r) }{r^2}  e^{3\lambda(r)} dr,
\end{equation}
for all $r\in[0,R)$, where $e^\lambda(r)=\left(1-\frac{2m(r)}{r}\right)^{-\frac{1}{2}}$.
\end{Thm}

Note that condition \eqref{necessary_cond_Thm} reduces to the classical Buchdahl condition when the density is constant. The proof of Theorem \ref{Thm_existence} is based on an explicit solution in integral form of the TOV-equation for staircase density functions, worked out in Section \ref{Sec_staircase-density}, obtained by a change of variables linearizing the TOV-equations. In Section \ref{Sec_proof_suff}, we complete the proof of Theorem \ref{Thm_existence}, by first approximating general density functions by ``staircase'' density functions, and by then applying the existence result of Section \ref{Sec_staircase-density} on sub-intervals in combination with a proof of convergence. In Section \ref{Sec_proof_nec} we give the proof of Theorem \ref{Thm_necessary}. In Section \ref{Sec_app_examples}, we discuss some examples exploring the degree of generality of conditions \eqref{Thm1} and \eqref{necessary_cond_Thm}.

\section{Exact Solutions for Staircase Density Functions}  \label{Sec_staircase-density}

The goal of this section is to solve the TOV-equation \eqref{TOV3} for ``staircase'' density functions. That is, we assume the density profile $\rho(r)$ is a step function which takes the form
\begin{equation}\label{density}
	\rho(r)=\begin{cases}
		\rho_i,\quad &r\in [R_{i-1},R_{i}),\\
		0,\quad &r\in [R_N,+\infty),\\
	\end{cases}
\end{equation}
where $\rho_i$ is a sequence of positive constants for each $i \in \{1,...,N\}$, ($N\in \mathbb{N}$ finite), and $R_i$ is a strictly increasing sequence of positive real numbers with $R_0=0$ (the stellar center) and $R_N=R$ (the stellar boundary). The goal now is to solve the TOV-equation \eqref{TOV3}, which decouples from the remaining Einstein-Euler equations \eqref{TOV1} - \eqref{TOV2}, on each interval $[R_{i-1},R_{i})$. We accomplish this by introducing a change of variables which linearizes the TOV-equations for each fixed $\rho_i$. By gluing the resulting solutions according to the RH-conditions \eqref{RH_p_cont}, we obtain a continuous solution $p$ of the TOV-equation on $[0,R]$ via a step-by-step procedure explained in the end of this section.    

The TOV-equation, which is a non-linear ODE of Ricatti-type, is difficult to solve in closed form even when the density is piece-wise constant; (for constant densities \eqref{TOV3} is known to be separable and solvable in closed form, but for a staircase density function separability fails, since $m(r)/r^3$ is no longer constant).  We resolve this obstacle by introducing the variable $\phi_i(r)$ on $[R_{i-1},R_i]$ as the harmonic mean
\begin{equation}\label{trans}
	\phi_i(r) \equiv \frac{p(r) \rho_i}{p(r)+\rho_i},
\end{equation}
by which we achieve a linearization of the original TOV equation as follows.  

\begin{Lemma} \label{Lemma_TOV-equivalence}
Assume $\rho_i >0$ is a constant density function on $[R_{i-1},R_i)$. \\ 
\noindent {\rm (i)} The TOV equation \eqref{TOV3} for $\rho_i$ on $[R_{i-1},R_i)$ is equivalent to 
	\begin{equation}\label{linearODE}
		\frac{d\phi_i(r)}{dr}+ \frac{4 \pi  r^3\rho_i -m(r)}{r (r-2m(r))} \phi_i(r)= \frac{-m(r) \rho_i  }{r(r-2m(r))}.
	\end{equation}
	That is, if $p$ is a solution of \eqref{TOV3}, then $\phi_i$ defined by \eqref{trans} solves \eqref{linearODE}; vice versa, if $\phi_i$ is a solution of \eqref{linearODE}, then $p$ defined by
	\begin{equation}\label{pressure0}
		p(r)=\frac {\rho_i\,\phi_i(r)}{\rho_i-\phi_i (r)}
	\end{equation}
solves the original TOV-equation \eqref{TOV3}. \\
\noindent {\rm (ii)} The general solution of the transformed TOV-equation \eqref{linearODE} on $[R_{i-1},R_i)$ is given by
	\begin{equation}\label{exact1}
		\phi_i(r)= e^{-\lambda(r)}
		\left(
		c_i + \int^{R_i}_r \frac{m\rho_i }{r^2}  \Big(1-\frac{2m}{r}\Big)^{-\frac{3}{2}} dr
		\right),
	\end{equation}
	where $c_i \in \mathbb{R}$ is some constant; and $p(r)$ defined by \eqref{pressure0} gives the corresponding solution of the original TOV-equation \eqref{TOV3}. 
\end{Lemma}

\begin{proof}
The TOV-equation \eqref{TOV3} for $\rho_i$ on $[R_{i-1},R_i)$ is given by 
\begin{equation} \label{TOV3_rho_i}
\frac{dp}{dr} = - \frac{\rho_i m}{r^2} \Big( 1 + \frac{p}{\rho_i} \Big)  \Big(1+ \frac{4\pi r^3 p}{m} \Big)  \Big( 1 - \frac{2m}{r} \Big)^{-1}.
\end{equation}
Now, to prove part (i), differentiate equation \eqref{pressure0} and substitute the resulting expression for $\frac{dp}{dr}$ into \eqref{TOV3_rho_i}, which gives
\[
\frac {\rho_i^2}{(\rho_i-\phi_i)^2} \frac{d\phi_i}{dr}
=
- \frac{1}{r(r-2m)} \left( 
\rho_i + \frac {\rho_i \phi_i}{\rho_i-\phi_i} \right)  \left(m+4\pi r^3 \frac {\rho_i \phi_i}{\rho_i-\phi_i} \right),
\]
from which we obtain \eqref{linearODE}. This proves the equivalence asserted in part (i).

We now prove part (ii). For the linear ODE  \eqref{linearODE},	the following is an integrating factor 
\begin{equation}\label{int_facttor}
	I(r)=\exp\left(\int  \frac{4 \pi  r^3 \rho_i-m}{r (r-2 m)} dr\right)=
	\left(1-\frac{2m}{r}\right)^{-\frac{1}{2}}=e^{\lambda(r)} , 
\end{equation}
and multiplying \eqref{linearODE} by its integrating factor $e^{\lambda(r)}$, we write \eqref{linearODE} equivalently as
$$
\frac{d}{dr}\Big( e^{\lambda(r)}\phi_i(r) \Big) = - e^{\lambda(r)} \frac{m\rho_i }{r^2}  \Big(1-\frac{2m}{r}\Big)^{-1}.
$$
Integration from $r$ to $R_i$ directly yields the exact solution \eqref{exact1} on $[R_{i-1}, R_i)$, as claimed. 
\end{proof}

Based on Lemma \ref{Lemma_TOV-equivalence}, we now construct a unique solution $p(r)$ of the TOV-equation for staircase density functions on $[0,R]$, such that $p(R)=0$ and such that $p(r)$ is Lipschitz continuous on $[0,R]$. For this, we use the RH-condition \eqref{RH_p_cont} to determine the free constants $c_i$ recursively, starting at the stellar surface $R=R_N$ and moving towards the center at $R_0=0$. To begin, note that $p(R_N)=0$, implies by \eqref{trans} that $\phi_N(R_N)=0$ and hence $c_N=0$. Thus, when $r\in[R_{N-1},R_N)$, 
\begin{equation}\label{exact3}
\phi_N (r)= e^{-\lambda(r)} \int^{R_N}_r \frac{m\rho_{N} }{r^2}  \left(1-\frac{2m}{r}\right)^{-\frac{3}{2}} dr.
\end{equation}
We now proceed recursively, moving from the outside to the inside. For this assume $\phi_{i+1}(r)$ is given for $r \in [R_{i},R_{i+1})$ and $i \in \{1,...,N-1\}$, we then choose $c_i$ as
\begin{equation} \label{choice_c_i}
	c_i  
= e^{\lambda(R_i)} \left(\frac{1}{\rho_{i}}+\frac{1}{\phi_{i+1}(R_{i})}-\frac{1}{\rho_{i+1}}\right)^{-1},
\end{equation}
which uniquely determines the next solution $\phi_{i}(r)$ for $r \in [R_{i-1},R_{i})$. The choice \eqref{choice_c_i} guarantees continuity of $p(r)$ across each interface $r=R_i$, because the Rankine Hugoniot jump conditions, $\lim_{r\to R_i^-}p(r)=\lim_{r\to R_i^+}p(r)$, in combination with \eqref{trans}, are equivalent to 
\begin{equation}\label{jump}
	\frac{1}{\phi_i(R_i)}-\frac{1}{\rho_i} 
	= \frac{1}{\phi_{i+1}(R_{i})}-\frac{1}{\rho_{i+1}},
\end{equation}
since $c_i= e^{\lambda(R_i)} \phi_i(R_i)$ by \eqref{exact1}, for each $i \in \{1,...,N-1\}$.  Taken on whole, setting for $r \in [0,R]$
\begin{equation}  \label{phiN_rhoN}
\phi(r) = \sum_{i=1}^{N} \phi_i (r) \chi_{i}(r)
\hspace{1cm} \text{and} \hspace{1cm}
\rho(r) = \sum_{i=1}^{N} \rho_i \chi_{i}(r),
\end{equation}
where $\chi_i(r)$ is the indicator function with $\chi_{i}(r) =1$ for $r\in[R_{i-1},R_i)$ and $\chi_{i}(r) =0$ otherwise, then
\begin{equation}\label{pressure}
p(r) = \frac{\rho(r)\phi(r)}{\rho(r)-\phi(r)},
\end{equation}
is the sought-after solution of the TOV-equations \eqref{TOV3} for the discontinuous staircase density function \eqref{density}, such that $p(R)=0$ and such that $p(r)$ is Lipschitz continuity on $[0,R]$. Moreover, the pressure function $p(r)$ defined in \eqref{pressure} is the only continuous solution of the TOV-equation \eqref{TOV3} for the staircase density function \eqref{density}, because the solutions of \eqref{linearODE} are unique by the Picard-Lindel\"off Theorem on each interval $[R_i,R_{i+1}]$, and there is only one way to match those solutions continuously.

\section{Sufficient Condition for a Bounded Pressure} \label{Sec_proof_suff}

In this section, we use the exact solution \eqref{pressure0} of the TOV-equation for staircase density functions, developed in Section \ref{Sec_staircase-density}, to prove existence of solutions $p(r)$ of the TOV-equation \eqref{TOV3} for general density functions containing a finite number of discontinuities (describing stellar phase transitions), thereby completing the proof of Theorem \ref{Thm_existence}. For this, assume $\rho(r)$ is piece-wise continuous, non-increasing and bounded on $[0,R]$, $\rho(r)>0$ for all $r\in [0,R)$ and $\rho(r)=0$ for all $r>R$, (so $\rho(R)$ may be non-zero). Moreover, assume \eqref{Outside_Schwarzschild}, i.e., 
\begin{equation} \label{Outside_Schwschild_Sec5} 
1 - \frac{2m(r)}{r} > K^2
\end{equation} 
for some constant $K \in (0,1)$; and assume our sufficient condition \eqref{Thm1}, i.e., there exists some constant $\Delta \in (0,1)$ such that    
\begin{equation} \label{delta}
	\int^{R}_0 \frac{m(r)}{r^2} \left(1-\frac{2m(r)}{r}\right)^{-\frac{3}{2}} dr \leq  \Delta .
\end{equation}
Then Theorem \ref{Thm_existence} asserts that there exists a unique solution $\{p(r),\nu(r)\} \in C^{0,1}([0,R])$ of the static spherically symmetric Einstein-Euler equations, \eqref{TOV1} - \eqref{TOV3}, satisfying the boundary data $p(R)=0$ and $e^{2\nu(R)} = \left(1-\frac{2m(R)}{R}\right)$. Moreover, the pressure function $p(r)$ is everywhere non-negative, strictly monotonously decreasing and bounded on $[0,R]$ with bound     
\begin{equation}
p(r) \ \leq \ \frac{\Delta }{e^{\lambda(r)}-\Delta}\rho(r) \ \leq \ \frac{\Delta}{1-\Delta} \rho(r)  .
\end{equation}

We now prove Theorem \ref{Thm_existence} by first establishing its claim for general staircase density profiles  in Lemma \ref{lemma1}; then, assuming $\rho(r)$ satisfies \eqref{delta}, we prove in Lemma \ref{lemma2} the existence of a sequence of staircase density functions $\{\rho^{(n)}(r)\}$, each satisfying condition \eqref{delta}, which approximates the ``physical'' density function $\rho(r)$; in the last step, in Lemmas \ref{lemma3} and \ref{lemma4}, we prove convergence of the approximate solutions $p^{(n)}(r)$, (i.e., solutions of the TOV-equation for $\rho^{(n)}(r)$), to the sought-after solution $p(r)$ of the TOV-equation for $\rho(r)$.

\begin{Lemma}\label{lemma1}
Let $\rho(r)=\sum_{i=1}^N \rho_i \chi_{i}(r)$ be any non-increasing staircase density function of form \eqref{density}, and let $m(r)$ be the associated mass function solving \eqref{TOV2}. Assume there exists some $\Delta\in(0,1)$ such that \eqref{delta} holds. Then the pressure function $p(r)$, defined in \eqref{pressure}, is the unique continuous solution of the TOV equation \eqref{TOV3} on $[0,R]$ with boundary data $p(R)=0$. Moreover,  $p(r)$ is bounded for all $r\in[0,R]$ by
	\begin{equation} \label{bound_p_Sec5}
		p(r) \ \leq \ \frac{\Delta }{e^{\lambda(r)}-\Delta}\rho(r) \ \leq \ \frac{\Delta}{1-\Delta} \rho(r).
	\end{equation} 
\end{Lemma}

\begin{proof}
The proof is based on the exact solutions $\phi_i(r)$, $r \in [R_{i-1},R_i]$, deduced in \eqref{exact1}. Note that we prove in Section \ref{Sec_staircase-density} that $p(r)$ defined in \eqref{pressure} is the sought-after unique solution of the TOV-equation with $p(R)=0$. It thus remains to prove the bound \eqref{bound_p_Sec5}; we do this by induction in $i =1,...,N$, starting from the outside $i=N$. For the case $i=N$, we directly obtain for the exact solution \eqref{exact3} that 
\begin{equation} \label{lemma_1_techeqn2}
\frac{\phi_N (r)}{\rho_N} \overset{\eqref{exact3}}{=}  
e^{-\lambda(r)}	\int^{R_{N}}_r \frac{m(\tau)}{\tau^2} e^{3\lambda(\tau)} d\tau 
\leq e^{-\lambda(r)} \Delta,
\end{equation}
where the last inequality follows from \eqref{delta}, keeping in mind that $e^{\lambda(\tau)} = \left(1-\frac{2m(r)}{r}\right)^{-\frac{1}{2}}$. Thus, since $\frac{1}{p(r)}+\frac{1}{\rho_N}=\frac{1}{\phi_N(r)}$ by \eqref{trans}, equation \eqref{lemma_1_techeqn2} directly implies on $[R_{N-1},R_{N}]$
$$
p(r) \ = \ \frac{\rho_{N}\phi_{N}(r)}{\rho_{N}-\phi_{N}(r)} 
\ \leq \ \frac{\Delta}{e^{\lambda(r)}-\Delta} \rho_N  \ \leq \ \frac{\Delta}{1-\Delta} \rho(r),
$$ 
where the last inequality follows since $e^{\lambda(r)} \geq 1 > \Delta$. 

Now, to continue the induction, we use that the jump condition \eqref{jump} at $r=R_{N-1}$ implies 
\begin{equation} \label{lemma1_techeqn3}
\frac{\phi_{N-1}(R_{N-1})}{\rho_{N-1}}\leq \frac{\phi_{N}(R_{N-1})}{\rho_{N}},
\end{equation} 
since $\rho_{N-1}\geq \rho_{N}>0$. Thus, for $i=N-1$, the exact solution \eqref{exact1} in combination with \eqref{lemma1_techeqn3} and \eqref{delta} yields
	\begin{align*}
		\frac{\phi_{N-1}(r)}{\rho_{N-1}}
		&=
		e^{-\lambda(r)}
		\left(
		e^{\lambda(R_{N-1})} \frac{\phi_{N-1}(R_{N-1}) }{{\rho_{N-1}}}+	\int^{R_{N-1}}_r \frac{m(\tau)}{\tau^2} e^{3\lambda(\tau)} d\tau 
		\right) \\
		&\overset{\eqref{lemma1_techeqn3}}{\leq}
		e^{-\lambda(r)}
		\left(
		e^{\lambda(R_{N-1})} \frac{\phi_{N}(R_{N-1})}{\rho_{N}} + \int^{R_{N-1}}_r \frac{m(\tau)}{\tau^2} e^{3\lambda(\tau)} d\tau 
		\right) \\
		&\leq
		e^{-\lambda(r)}
		\left(
		\int^{R_{N}}_{R_{N-1}} \frac{m(\tau)}{\tau^2} e^{3\lambda(\tau)} d\tau + \int^{R_{N-1}}_r \frac{m(\tau)}{\tau^2} e^{3\lambda(\tau)} d\tau 
		\right) \\
		&=e^{-\lambda(r)}
		\int^{R_{N}}_r \frac{m(\tau)}{\tau^2} e^{3\lambda(\tau)} d\tau \\
		&\overset{\eqref{delta}}{\leq} e^{-\lambda(r)} \Delta,
	\end{align*}
	for $r\in[R_{N-2},R_{N-1}]$. From this, we conclude again with the bound  
	$$
	p(r)=\frac{\rho_{N-1}\phi_{N-1}(r)}{\rho_{N-1}-\phi_{N-1}(r)}
	\leq \frac{\Delta}{e^{\lambda(r)}-\Delta}\rho_{N-1}
	$$
	on $[R_{N-2},R_{N-1}]$. 
	
To implement the induction step, we assume 
\begin{equation}  \label{lemma1_techeqn4}
\frac{\phi_k(R_k)}{\rho_k}\leq e^{-\lambda(R_k)} \int^{R_{N}}_{R_k} \frac{m(\tau)}{\tau^2} e^{3\lambda(\tau)} d\tau.
\end{equation}
By substitution, we then obtain
	\begin{align} \label{lemma1_techeqn5}
		\frac{\phi_k(r) }{\rho_k}
		&\overset{\eqref{exact1}}{=} 
		e^{-\lambda(r)}
		\left( e^{\lambda(R_k)}   \frac{\phi_k(R_k)  }{\rho_k }+	\int^{R_k}_r \frac{m(\tau)}{\tau^2}  e^{3\lambda(\tau)} d\tau	\right) \cr 
		&\overset{\eqref{lemma1_techeqn4}}{\leq}
		e^{-\lambda(r)}
		\left(
		\int^{R_{N}}_{R_{k}} \frac{m(\tau)}{\tau^2} e^{3\lambda(\tau)} d\tau + \int^{R_{k}}_r \frac{m(\tau)}{\tau^2} e^{3\lambda(\tau)} d\tau 
		\right) \cr 
		&=e^{-\lambda(r)}
		\int^{R_{N}}_r \frac{m(\tau)}{\tau^2} e^{3\lambda(\tau)} d\tau  \cr 
		& \leq e^{-\lambda(r)} \Delta,
	\end{align}
	for $r \in [R_{k-1}, R_k]$. From this, we obtain again for $r \in [R_{k-1}, R_k]$ the pressure bound 
	\begin{equation} \label{lemma1_techeqn6}
	p(r) = \frac{\rho_{k}\phi_{k}(r)}{\rho_{k}-\phi_{k}(r)} \leq \frac{\Delta}{e^{\lambda(r)}-\Delta}\rho_{k} .
	\end{equation}
	Moreover, the jump condition \eqref{jump} at $r=R_{k-1}$ yields
	\[
	\frac{\phi_{k-1}(R_{k-1})}{\rho_{k-1}}\leq \frac{\phi_{k}(R_{k-1})}{\rho_{k}}
	\overset{\eqref{exact1}}{=} e^{-\lambda(R_{k-1})}
	\int^{R_{N}}_{R_{k-1}} \frac{m(\tau)}{\tau^2} e^{3\lambda(\tau)} d\tau,
	\]	
	which verifies the induction assumption \eqref{lemma1_techeqn4} at the next level. We thus conclude that \eqref{lemma1_techeqn4}, \eqref{lemma1_techeqn5} as well as \eqref{lemma1_techeqn6} hold for all $k \in \{1,...,N\}$. Moreover, the pressure function $p(r)$ defined by \eqref{pressure} satisfies \eqref{bound_p_Sec5}, as claimed. This completes the proof.
\end{proof}

In the next lemma, we show that any bounded density function with at most finitely many jump discontinuities, satisfying \eqref{delta}, can be approximated by staircase density functions, each satisfying \eqref{delta}.

\begin{Lemma}\label{lemma2}
Suppose $\rho(r)$ is a non-increasing bounded function defined on $[0, R]$ with at most finitely many jump discontinuities (for which both left- and right-sided limits exist), such that \eqref{delta} holds. Then there exists a sequence of staircase density functions $\{\rho^{(n)}(r)\}$ of form \eqref{density}, that converges to $\rho(r)$ in $L^{\infty}([0, R])$. Moreover, for $m^{(n)}(r)=\int_{0}^{r} 4\pi \tau^2 \rho^{(n)}(\tau) d\tau$ and $\rho^{(n)}(\tau)$ as in \eqref{phiN_rhoN}, $r\in[0,R]$, conditions \eqref{Outside_Schwschild_Sec5} and \eqref{delta} hold in the sense that
\begin{equation} \label{approx}
	1 - \frac{2m^{(n)}(r)}{r} > K^2
	\hspace{1cm} \text{and} \hspace{1cm}
	\int^{R}_0 \frac{m^{(n)}(r)}{r^2}  \left(1-\frac{2m^{(n)}(r)}{r}\right)^{-\frac{3}{2}} dr \ \leq \ \Delta.
\end{equation}
\end{Lemma}

\begin{proof}
We now develop an explicit approximation scheme based on halving subintervals of $[0,R]$ consecutively. For this let $R^{(0)}_j$ denote the position of the jump discontinuities of $\rho(r)$ in $[0,R]$, for  $j=1,\cdots,J$ finite, which gives a partition $\mathcal{P}_0: 0=R^{(0)}_0<R^{(0)}_1<\cdots<R^{(0)}_{J+1}=R$ of $[0,R]$. Clearly $\rho(r)$ is continuous on the sub-intervals $(R^{(0)}_{j-1},R^{(0)}_{j})$ for each $j=1,\cdots,J+1$, and both right-sided limits $\rho(R^{(0)}_{j-1} +)=\lim_{r\searrow R^{(0)}_{j-1}}\rho(r)$ and left-sided limits $\rho(R^{(0)}_{j} -)=\lim_{r\nearrow R^{(0)}_{j}}$ exist by assumption. 	

To construct a sequence of approximating staircase density functions, we begin by adding midpoints $\frac{1}{2}(R^{(0)}_{j-1} + R^{(0)}_{j})$ for each $j=1,\cdots,J,J+1$ to the initial partition $\mathcal{P}_0$,  which yields a refined partition $\mathcal{P}_1: 0=R^{(1)}_0<R^{(1)}_1<\cdots<R^{(1)}_{2J+2}=R$. We now define the first associated staircase density function on $[0,R]$ by setting
	\begin{equation}
		\rho^{(1)}(r) = \rho(R^{(1)}_j -), \text{ if } r\in[R^{(1)}_{j-1}, R^{(1)}_{j}) \text{ for some } j\in \{1,\cdots,2J+2\},
	\end{equation}
	and by setting $\rho^{(1)}(r)=\rho(R^{(1)}_{2J+2}-)$, if $r=R$.
	
We now repeat the procedure of adding midpoints to each subinterval of each partition iteratively. That is, for each interval $(R^{(n-1)}_{j-1} , R^{(n-1)}_{j})$ of the partition $\mathcal{P}_{n-1}$, we add the midpoint $\frac{1}{2}(R^{(n-1)}_{j-1} + R^{(n-1)}_{j})$, to obtain a new refined partition $\mathcal{P}_{n}$. After n-th iterations, the process gives rise to a partition 
$$
\mathcal{P}_n: 0=R^{(n)}_0<R^{(n)}_1<\cdots<R^{(n)}_{2^n( J+ 1)}=R
$$ 
and a staircase density function on $[0,R]$ defined by
	\begin{equation}\label{partition}
		\rho^{(n)}(r) = \rho(R^{(n)}_j -), \text{   if } r\in[R^{(n)}_{j-1}, R^{(n)}_{j}) \text{ for some } j\in \big\{1,\cdots, 2^n( J+ 1)\big\},
	\end{equation}
	and $\rho^{(n)}(r)=\rho(R^{(n)}_{2^n( J+ 1)}-)$, if $r=R$.
	
Next, we prove the sequence of staircase density functions $\rho^{(n)}(r)$ constructed in \eqref{partition} convergences uniformly to $\rho(r)$ almost everywhere. For this, we extend $\rho(r)$ continuously to the closed intervals $[R_{j-1}^{(0)},R_{j}^{(0)}]$. Then, by the Heine--Cantor Theorem \cite{Rudin}, for each $j = 1,..,J+1$ the so extended density function $\rho$ is uniformly continuous on $[R^{(0)}_{j-1},R^{(0)}_{j}]$, in the sense that for any $\varepsilon>0$, there exists $\delta_j>0$ such that 
\begin{equation} \label{uniform_convergence}
|\rho(x)-\rho(y)| < \varepsilon \ \ \ \ \ \forall x,y\in [R^{(0)}_{j-1},R^{(0)}_{j}] \ \ \ \text{with} \  \ \ |x-y|<\delta_j.
\end{equation}
Taking $\delta \equiv \min \{\delta_j \mid j=1,..., J+1\}$, we choose some $N\in\mathbb{N}$ such that 
\[
\max_{1\leq j\leq 2^N (J+1)} |R^{(N)}_{j}- R^{(N)}_{j-1}| <\delta .
\]
Then, for any $n >N$ and for any $r \in [0,R]$ with $r \in [R^{(n)}_{j-1}, R^{(n)}_{j})$, it follows that
\begin{equation} \label{lemma2_sup-norm}
|\rho(r)-\rho^{(n)}(r)| \overset{\eqref{partition}}{=} |\rho(r)-\rho(R_j^{(n)}-)| \overset{\eqref{uniform_convergence}}{<} \varepsilon .
\end{equation}
Thus, since the above choice of $n>N$ is independent of $r \in [0,R]$, we conclude that 
\[
\|\rho-\rho^{(n)}\|_{L^\infty([0,R])} 
= \max_{j=1,...,2^n(J+1)} \|\rho-\rho^{(n)}\|_{L^\infty([R^{(n)}_{j-1}, R^{(n)}_{j}))} 
\overset{\eqref{lemma2_sup-norm}}{<} \varepsilon.
\] 
This implies that $\rho^{(n)}(r)$ converges to $\rho(r)$ in $L^\infty([0,R])$ as $n \to \infty$, just as claimed.

Finally, to prove \eqref{approx}, observe that by \eqref{partition} we have $\rho^{(n)}(r)\leq\rho(r)$ for almost everywhere $r\in[0,R]$. Thus, since inequalities are preserved by integration, we obtain 
$$
m^{(n)}(r)=\int_{0}^{r} 4\pi \tau^2 \rho^{(n)}(\tau) d\tau\leq \int_{0}^{r} 4\pi \tau^2 \rho(\tau) d\tau=m(r)
$$ 
and 
$$
e^{\lambda_n(r)}= \left(1-\frac{2m^{(n)}(r)}{r}\right)^{-\frac{1}{2}} \leq \left(1-\frac{2m(r)}{r}\right)^{-\frac{1}{2}} = e^{\lambda(r)},
$$ 
from which we deduced the sought-after inequality 
$$ 
1 - \frac{2m^{(n)}(r)}{r} \geq 1 - \frac{2m(r)}{r} > K^2    
$$
and
$$
\int^{R}_0 \frac{m^{(n)}(r)}{r^2}  e^{3\lambda_n(r)} dr \leq 	\int^{R}_0 \frac{m(r)}{r^2}  e^{3\lambda(r)} dr \leq \Delta.
$$
This shows that each $\rho^{(n)}(r)$ satisfies \eqref{approx}, which completes the proof.
\end{proof}

To obtain the convergence of the pressure $p^{(n)}(r)$, we show the convergence of $\phi^{(n)}(r)$ as an intermediate step.

\begin{Lemma}\label{lemma3}
Assume $\rho(r)$ is a density functions as in Lemma \ref{lemma2}, and assume $\{\rho^{(n)}(r)\}$ is the corresponding sequence of staircase density functions as defined in \eqref{partition}, such that $\{\rho^{(n)}(r)\}$ converges to $\rho(r)$ in $L^{\infty}([0, R])$. For each $\rho^{(n)}(r)$, let $\phi^{(n)}(r)$ be the solution of the of the piece-wise linearized TOV-equation \eqref{linearODE} associated to $\rho^{(n)}(r)$, as defined in \eqref{phiN_rhoN}.\footnote{More precisely, we take $\phi^{(n)}(r)\equiv \phi(r)$ for $\phi(r)$ as in \eqref{phiN_rhoN} with $N=2^n(J+1)$ and $\rho^{(n)}(r)$ the staircase density function defined in \eqref{partition}, and where $J$ is the number of discontinuities in $\rho(r)$.} Then the sequence $\{\phi^{(n)}(r)\}$ converges in $L^{\infty}([0, R])$.
\end{Lemma}
\begin{proof}
	By convergence of $\{\rho^{(n)}(r)\}$ in $L^{\infty}([0, R])$, it follows that for any $\varepsilon>0$, there exists an $N \in \mathbb{N}$ such that 
	\begin{equation}\label{lemma3_techeqn1}
		|\rho^{(i)} (r) - \rho^{(j)} (r)|<\varepsilon, \quad \forall i,j>N,
	\end{equation}
	for almost everywhere $r\in[0,R]$, (so $N$ is independent of $r$).
	From \eqref{lemma3_techeqn1} we obtain for any $i,j >N$ that
	\begin{equation}\label{lemma3_techeqn2}
		|m^{(i)}(r)-m^{(j)}(r)| 
		= 
		\left|\int_0^r 4\pi \tau^2 (\rho^{(i)}(\tau)-\rho^{(j)}(\tau) )d\tau\right| \\
		\leq
		\int_0^r 4\pi \tau^2 |\rho^{(i)}(\tau)-\rho^{(j)}(\tau)| d\tau \\
		\leq
		\frac{4}{3}\pi r^3 \varepsilon.
	\end{equation}
	For the metric component $e^{-2\lambda_n}=1-\frac{2m^{(n)}(r)}{r}$, we find that
	\begin{equation}\label{lemma3_techeqn3}
			|e^{-2\lambda_i } - e^{-2\lambda_j }|
			= 
			\frac{2}{r}\left| m^{(i)}(r)-m^{(j)}(r) \right| 
			\overset{\eqref{lemma3_techeqn2}}{\leq} 
			\frac{8}{3}\pi r^2 \varepsilon
			\leq
			\frac{8}{3}\pi R^2 \varepsilon
	\end{equation}
	and, since by \eqref{approx} we have $e^{\lambda_n} < K^{-1}$, we get in addition
	\begin{equation}\label{lemma3_techeqn4}
		| e^{2\lambda_i } - e^{2\lambda_j } |
		= 
		\left|e^{2\lambda_i}e^{2\lambda_j}\right| \left|e^{-2\lambda_i}-e^{-2\lambda_j}\right| 
		\overset{\eqref{approx}}{\leq} 
		\frac{1}{K^4} \left|e^{-2\lambda_i}-e^{-2\lambda_j}\right| 
		\overset{\eqref{lemma3_techeqn3}}{\leq} 
		\frac{8\pi R^2}{3 K^4} \varepsilon.
	\end{equation}
	Since $e^{\lambda_n(r)}\geq 1$ for all $r\in[0,R]$, we obtain 
	\begin{equation}\label{lemma3_techeqn5}
		| e^{\lambda_i} - e^{\lambda_j} | 
		=
		\left|\frac{e^{2\lambda_i}-e^{2\lambda_j}}{e^{\lambda_i} + e^{\lambda_j} }\right| 
		\leq
		\frac{1}{2} \left|e^{2\lambda_i} - e^{2\lambda_j}\right|
		\overset{\eqref{lemma3_techeqn4} }{\leq}
		\frac{4\pi R^2}{3 K^4} \varepsilon.
	\end{equation}
	Combining now \eqref{lemma3_techeqn4} and \eqref{lemma3_techeqn5}, we get
	\begin{align}\label{lemma3_techeqn6}
		|e^{3\lambda_i} - e^{3\lambda_j}|
		= 
		|e^{\lambda_i}e^{2\lambda_i} - e^{\lambda_j} e^{2\lambda_j}| 
		&\leq
		|e^{\lambda_i}| |e^{2\lambda_i} - e^{2\lambda_j}| + |e^{2\lambda_j}| | e^{\lambda_i} - e^{\lambda_j} |  \cr
		&{\leq}
		\frac{1}{K} |e^{2\lambda_i} - e^{2\lambda_j}| + \frac{1}{K^2}| e^{\lambda_i} - e^{\lambda_j} |  
		\ \leq \
		\frac{4\pi R^2 }{K^6} \varepsilon,
	\end{align}
	keeping in mind that $0<K<1$.

The goal now is to prove convergence of $\{\phi^{(n)}(r)\}$ in $L^\infty([0,R])$. 
For this, consider two such linear solutions $\phi^{(i)}(r)$ and $\phi^{(j)}(r)$. We then denote the common refinement of partitions corresponding to $\rho^{(i)}(r)$ and $\rho^{(j)}(r)$ as $0=R_0<R_1<\cdots<R_k=R$, (that is, if $j\geq i$, then $k=2^j(J+1)$). One can show that both $\phi^{(i)}(r)$ and $\phi^{(j)}(r)$ satisfy, for any $r\in[R_{l-1},R_l]$, any $l=1,\cdots,k$, that  
	\begin{align}\label{lemma3_techeqn7}
	e^{\lambda_n(r)} 	\phi^{(n)} (r) 
	= 
	\int^{R_l}_{r} \frac{m^{(n)}(\tau)}{\tau^2}\rho^{(n)}(\tau)e^{3\lambda_n(\tau)} d\tau  + e^{\lambda_n(R_l)} \phi^{(n)} (R_l) 
	=: f^{(n)}_l(r)  .
	\end{align}
To prove convergence of $\{ \phi^{(n)} (r) \}$ in $L^\infty([0,R])$, it suffices to prove convergence of $\{ f^{(n)}_l \}$ in $L^\infty([0,R])$ for each $l=1,...,k$, since
	\begin{align} \label{lemma3_techeqn7'}
	\|\phi^{(i)}-\phi^{(j)}\|_{L^{\infty}([0,R])}
	= &
	\max_{1\leq l \leq k} \| e^{-\lambda_i} f^{(i)}_l  -  e^{-\lambda_j} f^{(j)}_l \|_{L^{\infty}([R_{l-1},R_{l}])} 	\cr 
	\leq &
	\max_{1\leq l \leq k}	\| f^{(i)}_l - 	f^{(j)}_l \|_{L^{\infty}([R_{l-1},R_{l}])}
	+
	\|f^{(i)}_l\|_{L^{\infty}([0,R])} \;  \| e^{-\lambda_i} - e^{-\lambda_j}\|_{L^{\infty}([0,R])},  \ \ \ \ \ \ \
	\end{align}
which, by \eqref{lemma3_techeqn3}, converges to zero as $i,j \to \infty$, once the convergence of $\{ f^{(n)}_l \}$ has been established. 

To prove this convergence, we estimate both terms on the right hand side of \eqref{lemma3_techeqn7} separately, starting with the first one. We first write                 
	\begin{align} \nonumber
		&\left| \int^{R_l}_{r} \frac{m^{(i)}}{\tau^2}\rho^{(i)}e^{3\lambda_i} d\tau
		-\int^{R_l}_{r} \frac{m^{(j)}}{\tau^2}\rho^{(j)}e^{3\lambda_j} d\tau\right| 
		\leq
		\int^{R_l}_{r}  
		\left|	
		\frac{m^{(i)}}{\tau^2}\rho^{(i)}e^{3\lambda_i} 
		-
		\frac{m^{(j)}}{\tau^2}\rho^{(j)}e^{3\lambda_j} 
		\right|
		d\tau \cr
		=&
		\int^{R_l}_{r}  \left| \tfrac{\rho^{(i)}e^{3\lambda_i(r)}}{r^2} \Big(m^{(i)}(r)-m^{(j)}(r)\Big)  
		+ \tfrac{\rho^{(i)}m^{(j)}(r)}{r^2} \Big(e^{3\lambda_i(r)} - e^{3\lambda_j(r)}\Big)
		+ \tfrac{m^{(j)}(r)e^{3\lambda_j(r)}}{r^2} \Big(\rho^{(i)}(r)-\rho^{(j)}(r)\Big)  \right|dr ,
	\end{align}
and substituting \eqref{lemma3_techeqn2} into the first term, and substituting \eqref{lemma3_techeqn6} and \eqref{lemma3_techeqn1} together with the bound $|m^{(i)}(r)| \leq C r^3$ into the second and third term, we obtain in terms of a universal constant $C>0$ that
	\begin{align} \label{lemma3_techeqn8}
		\left| \int^{R_l}_{r} \frac{m^{(i)}}{\tau^2}\rho^{(i)}e^{3\lambda_i} d\tau
		-\int^{R_l}_{r} \frac{m^{(j)}}{\tau^2}\rho^{(j)}e^{3\lambda_j} d\tau\right|
		\leq & \ C \int^{R_l}_{0}  
		\left(
		\frac{4}{3}\pi R 
		+ \frac{4\pi R^2}{K^6}
		+ 1
		\right)\varepsilon  dr
		\ \leq \ 
		C \varepsilon,
	\end{align}	
which is the sought-after estimate on the first term.

We estimate the second term in \eqref{lemma3_techeqn7} by induction.
	When $l=k$, the boundary condition $p(R_k)=0$ implies $\phi^{(n)} (R_k)=0$.
	Hence $f^{(n)}_k (r)$ is uniformly convergent on $[R_{k-1},R_k]$ and $|f^{(i)}_k (R_{k-1})-f^{(j)}_k (R_{k-1})| \leq C\varepsilon$ by \eqref{lemma3_techeqn7} and \eqref{lemma3_techeqn8}. Now we prove $f^{(n)}_{l-1} (r)$ is uniformly convergent on $[R_{l-2},R_{l-1}]$ by induction, assuming that $f^{(n)}_{l} (r)$ is uniformly convergent on $[R_{l-1},R_{l}]$. Note that $\phi^{(n)}(r)$ is discontinuous at $r=R_{l-1}$, so the Rankine--Hugoniot conditions, denoting the jump in a function $f$ across $R_l$ by $[f]_l\equiv f(R_l^+) - f(R_l^-)$,  read
	\begin{equation}\label{lemma3_techeqn9}
		\left[	\frac{1}{\phi^{(n)}} \right]_{l-1} = \left[\frac{1}{\rho^{(n)}}\right]_{l-1}.
	\end{equation}
	Multiplying both sides of \eqref{lemma3_techeqn9} by $e^{-\lambda_n(R_{l-1})}$ yields
	\begin{equation}\label{lemma3_techeqn10}
		\frac{1}{f_{l-1}^{(n)}(R_{l-1})} 
		=
		 \frac{1}{f_l^{(n)}(R_{l -1})}
		 -
		 e^{-\lambda_n(R_{l-1})}
		 \left[\frac{1}{\rho^{(i)}}\right]_{l-1}.
	\end{equation}
	By our induction assumption on uniform convergence of $f^{(n)}_{l}(r)$ for $r \in [R_{l-1},R_{l}]$, we have
	\begin{align*}
		\left|\frac{1}{f_l^{(i)}(R_{l-1})}-\frac{1}{f_l^{(j)}(R_{l-1})}\right| 
		= 
		\frac{|f_l^{(i)}(R_{l-1}) - f_l^{(j)}(R_{l-1})|}{f_l^{(i)}(R_{l-1})\, f_l^{(j)}(R_{l-1})} \leq C \varepsilon,
	\end{align*}
and by uniform convergence of $\{\rho^{(n)}(r)\}$, we have
	\begin{align*}
		\left|
		\left[\frac{1}{\rho^{(i)}}\right]_{l-1}-\left[\frac{1}{\rho^{(j)}}\right]_{l-1} 
		\right|
		\leq
		\frac{\left|\rho^{(i)}(R_{l-1}^+) -\rho^{(j)}(R_{l-1}^+)\right|}{\rho^{(i)}(R_{l-1}^+) \rho^{(j)}(R_{l-1}^+)}
		+\frac{\left|\rho^{(j)}(R_{l-1}^-)- \rho^{(i)}(R_{l-1}^-)\right|}{\rho^{(j)}(R_{l-1}^-) \rho^{(i)}(R_{l-1}^-)}
		\leq  C \varepsilon,
	\end{align*}
	and combining the above we obtain
	\begin{align}\label{lemma3_techeqn11}
		 |f_{l-1}^{(i)}(R_{l-1})&-f_{l-1}^{(j)}(R_{l-1})| 
		 =  \left|f_{l-1}^{(i)}(R_{l-1}) \right|  \left|f_{l-1}^{(j)}(R_{l-1})\right| 
		\left| \frac{1}{f_{l-1}^{(i)}(R_{l-1})} -\frac{1}{f_{l-1}^{(j)}(R_{l-1})} \right| 
		\cr
		\overset{\eqref{lemma3_techeqn10}}{\leq} &
		C_1 
		\left|\frac{1}{f_l^{(i)}(R_{l-1})}-\frac{1}{f_l^{(j)}(R_{l-1})}\right| 
		+
		C_2	
		| e^{-\lambda_i (r)} - e^{-\lambda_j (r)} | 
		+
		C_3
		\left|
		\left[\frac{1}{\rho^{(i)}}\right]_{l-1}-\left[\frac{1}{\rho^{(j)}}\right]_{l-1}
		\right| 
		\cr
		\leq &  C \varepsilon.
	\end{align}
	Combining now \eqref{lemma3_techeqn8} and \eqref{lemma3_techeqn11} yields
	\begin{align*}
		|f^{(i)}_{l-1} (r) - 	f^{(j)}_{l-1}(r)| 
		\leq
		|f_{l-1}^{(i)}(R_{l-1})-f_{l-1}^{(j)}(R_{l-1})| 
		+
		\left|
		\int^{R_{l-1}}_{r} \frac{m^{(i)}}{\tau^2}\rho^{(i)}e^{3\lambda_i} d\tau
		-\int^{R_{l-1}}_{r} \frac{m^{(j)}}{\tau^2}\rho^{(j)}e^{3\lambda_j} d\tau\right| 
		\leq C \varepsilon,
	\end{align*}
	which implies the sought-after uniform convergence of $f^{(n)}_{l-1} (r)$ on $[R_{l-2},R_{l-1}]$. Applying the procedure for $l=k-2,\cdots,1$, it follows that $f^{(n)}_{l} (r)$ is uniformly convergent on $[R_{l-1},R_{l}]$ for any $l=1,...,k$. This, in combination with \eqref{lemma3_techeqn7'}, completes the proof. 
\end{proof}

We are now prepared to prove convergence of approximate pressure functions $p^{(n)}(r)$, associated to each staircase density function $\rho^{(n)}(r)$, to a pressure function solving the TOV-equation for the original density function $\rho(r)$. The proof of Theorem \ref{Thm_existence} is then a direct consequence of the following lemma.

\begin{Lemma}\label{lemma4}
	Assume $\rho(r)$ is a density functions as in Lemma \ref{lemma2}, and assume $\{\rho^{(n)}(r)\}$ is the corresponding sequence of staircase density functions as defined in \eqref{partition}, such that $\{\rho^{(n)}(r)\}$ converges to $\rho(r)$ in $L^{\infty}([0, R])$. For each $n\in \mathbb{N}$, let $p^{(n)}(r)$ be the solution of the TOV equation \eqref{TOV3} associated to $\rho^{(n)}(r)$ with $p^{(n)}(R)=0$ as introduced in \eqref{pressure}. Then each $p^{(n)}(r)$ is Lipschitz continuous, and the sequence $\{p^{(n)}(r)\}$ converges to a Lipschitz continuous function $p(r)$ on $[0,R]$ with respect to the Lipschitz norm, which solves the TOV-equation \eqref{TOV3} for the original $\rho(r)$ and which satisfies the uniform bound
\begin{equation} \label{bound_p-Lemma4}
p(r) \ \leq \ \frac{\Delta }{e^{\lambda(r)}-\Delta}\rho(r) \ \leq \ \frac{\Delta}{1-\Delta} \rho(r).
\end{equation}
\end{Lemma}

\begin{proof}
	By convergence of $\{\rho^{(n)}(r)\}$ and $\{\phi^{(n)}(r)\}$ in $L^{\infty}([0, R])$, it follows that for any $\varepsilon>0$, there exists an integer $N$ independent of $r$ such that 
	\begin{equation} \label{lemma4_techeqn1} 
		|\rho^{(i)} (r) - \rho^{(j)} (r)|<\varepsilon
		\hspace{1cm} \text{and} \hspace{1cm}
		|\phi^{(i)} (r) - \phi^{(j)} (r)|<\varepsilon
	\end{equation}
	for all $i,j>N$, almost everywhere $r \in [0,R]$. 
	Recall, by \eqref{pressure}, the pressure function for $r\in[0,R]$ takes the form
	\begin{align*}
		p^{(n)}(r)=\frac{\rho^{(n)}(r) \phi^{(n)}(r)}{\rho^{(n)}(r)-\phi^{(n)}(r)}.
	\end{align*}
	 Then we obtain, by combining the above formula for $p^{(n)}(r)$ with \eqref{lemma4_techeqn1}, that
	\begin{align}
	 	|p^{(i)}(r)-&p^{(j)}(r)|  
		\leq 
		\frac{| \rho^{(i)}(r) \rho^{(j)}(r) |\   |\phi^{(i)}(r) - \phi^{(j)}(r)| +  | \phi^{(i)}(r) \phi^{(j)} (r) | \   |\rho^{(i)}(r) - \rho^{(j)}(r) | }{|\rho^{(i)} (r) - \phi^{(i)} (r) |\  |\rho^{(j)} (r) -\phi^{(j)}(r) |}  
		\cr 
		\leq & 
		\left|1 -\frac{ \phi^{(i)} (r) }{\rho^{(i)} (r) }\right|^{-1} 
		\left|1 -\frac{ \phi^{(j)} (r) }{\rho^{(j)} (r) }\right|^{-1} 
		\left(|\phi^{(i)}(r) - \phi^{(j)}(r)| +  \left| \frac{\phi^{(i)}(r) }{\rho^{(i)}(r)}\right| \left| \frac{\phi^{(j)}(r) }{\rho^{(j)}(r)}\right|  |\rho^{(i)}(r) - \rho^{(j)}(r) | \right)
		\cr 
\leq & C \epsilon.   \label{lemma4_techeqn3}
	\end{align}
	
Now, recall that $p^{(n)}$ is Lipschitz continuous on $[0,R]$, as a result of the matching construction leading to \eqref{pressure}. By \cite[Ch.5]{evans}, the Lipschitz norm is equivalent to the Sobolev norm $W^{1,\infty}$, defined by 
		\begin{equation}\label{Lipschitz}
			\|p^{(n)}(r)\|_{W^{1,\infty}}:=\|p^{(n)}(r)\|_{L^\infty}+\left|\left|\frac{d}{dr}p^{(n)}(r)\right|\right|_{L^\infty},
		\end{equation} 
in terms of the $L^\infty$ norm $\|\cdot\|_{L^\infty}$ on $[0,R]$. We now estimate the $W^{1,\infty}$-norm on differences of sequence elements $\{p^{(n)}\}$ as follows: 
By	the TOV equation \eqref{TOV3} we find for any $r\in [0,R]$ that
	\begin{align*}
			\left|\frac{dp^{(i)}}{dr}-\frac{dp^{(j)}}{dr}\right| 
			=&
			\left|
			(\frac{m^{(i)}}{r^2}+4\pi r p^{(i)}) 
			(p^{(i)}+\rho^{(i)})
			e^{2\lambda_i }
			-
			(\frac{m^{(j)}}{r^2}+4\pi r p^{(j)}) 
			(p^{(j)}+\rho^{(j)})
			e^{2\lambda_j }
			\right|
			\\
			\leq &  
			C \Big( \left|\frac{m^{(i)}}{r^2}-\frac{m^{(j)}}{r^2}\right| + \left|(p^{(i)}-p^{(j)})\right|
			+ \left|(\rho^{(i)}-\rho^{(j)}) \right|
			+
			 \left| e^{2\lambda_i }-e^{2\lambda_j } \right|   \Big)
			\leq  C \epsilon,
		\end{align*}
where the last inequality follows by \eqref{lemma4_techeqn3} as well as \eqref{lemma3_techeqn1} - \eqref{lemma3_techeqn3}.	Therefore, $\{p^{(n)}(r)\}$ is a Cauchy sequence in $W^{1,\infty}([0,R])$ and hence converges to some Lipschitz continuous function $p(r)$ on $[0,R]$. Note finally that by \eqref{bound_p_Sec5} of Lemma \ref{lemma1}, we have for each $r\in [0,R]$ the upper bound
$$
p(r) \ \leq \ \frac{\Delta }{e^{\lambda_n(r)}-\Delta}\rho^{(n)}(r) \ \leq \ \frac{\Delta }{e^{\lambda(r)}-\Delta}\rho(r),
$$ 
by construction of $\rho^{(n)}(r)$ in \eqref{partition}. That $p(r)$ solves the TOV-equation \eqref{TOV3} for the original $\rho(r)$ follows directly by the above Lipschitz convergence of the approximate solutions $p^{(n)}(r)$ to $p(r)$.  This completes the proof.
	\end{proof}

\section{Necessary Condition}  \label{Sec_proof_nec}

In \cite{Buchdahl}, Buchdahl shows that in order to avoid pressure blow-up at the origin $r=0$, the stellar radius $R$ and its total mass $M$ must satisfy $\frac{2M}{R}<\frac{8}{9}$, cf. \eqref{Buchdahl}. In Theorem \ref{Thm_necessary} we generalize Buchdahl's condition \eqref{Buchdahl} and provide a new necessary condition for the existence of a finite central pressure.  We now prove Theorem \ref{Thm_necessary}. For this assume $\rho(r)$ is a bounded, non-increasing, non-negative, piece-wise differentiable density function with at most finitely many jump discontinuities on $[0,R]$, such that $\rho(r)=0$ for all $r>R$, and such that \eqref{Outside_Schwarzschild} holds for all $r\in[0,R]$. Assume $p(r)$ is a non-negative Lipschitz continuous solution of the TOV equation \eqref{TOV3} with $p(R)=0$ and which is bounded on $[0,R]$. Then Theorem \ref{Thm_necessary} asserts that the density $\rho(r)$ satisfies for all $r\in[0,R)$ the (necessary) condition   	
\begin{equation}\label{necessary2}
		\rho(r)> e^{-\lambda(r)}\int^{R}_r \frac{m(r)\rho(r) }{r^2}  e^{3\lambda(r)} dr.
\end{equation}

\noindent {\it Proof of Theorem \ref{Thm_necessary}.}
	The jump discontinuities form a partition of $[0,R]$ given by $0=R_0<R_1<\cdots<R_N=R$. So $\rho(r)$ is differentiable on $(R_{k-1},R_k)$ for every $1\leq k \leq N$. 
	Now, let $r \in [R_{k-1},R_{k})$ for some fixed $k \in \{ 1,..., N\}$, and define 
\begin{equation} \label{phi_sec6}
\Phi(r)=\frac{p(r) \rho(r)}{p(r)+\rho(r)}  \geq 0.
\end{equation} 
Substituting $\Phi(r)$ into the TOV equation \eqref{TOV3} yields
	\begin{equation}
		\Phi'(r) =  \frac{\rho'(r)}{\rho^2(r)}\Phi^2(r) +\frac{m(r)\Phi(r)-4 \pi  r^3\rho(r)\Phi(r) - m(r)\rho(r) }{r(r-2m(r))}.
	\end{equation}
	Since $\rho'(r)\leq0$ for any $r\in[R_{k-1},R_k)$, we obtain
	\begin{equation}
		\Phi'(r) \leq \frac{m(r) - 4 \pi  r^3\rho(r) }{r (r-2m(r))} \Phi(r) - \frac{m(r) \rho(r)  }{r(r-2m(r))},
	\end{equation}
	which can be rewritten as 
	\begin{equation}
		\frac{d}{dr}\left( \Phi(r)  e^{\lambda(r)}\right) 
		\leq
		- \frac{m(r) \rho(r)  }{r^2} e^{3\lambda(r)}.
	\end{equation}
	Integrating both sides implies
	\begin{equation} \label{proof_nec_eqn1}
		\int_r^{R_k} \frac{m(\tau)\rho(\tau)}{\tau^2} e^{3\lambda(\tau)} d\tau \leq \Phi(r)  e^{\lambda(r)} -\Phi(R_k^-)  e^{\lambda(R_k^-)}.
	\end{equation}
	Similarly,  for $i= k,\cdots,N-1$, we have
	\begin{equation} \label{proof_nec_eqn2}
		\int^{R_{i+1}}_{R_{i}} \frac{m(\tau)\rho(\tau)}{\tau^2} e^{3\lambda(\tau)} d\tau
		\leq
		\Phi(R_{i}^+)  e^{\lambda(R_{i}^+)} -\Phi(R_{i+1}^-)  e^{\lambda(R_{i+1}^-)}.
	\end{equation}	 
Thus, adding \eqref{proof_nec_eqn1} to \eqref{proof_nec_eqn2} for all $i\in \{k,...,N-1\}$, and using that $\Phi(R_N)=0$ and that $e^\lambda(r)$ is continuous, we obtain 
	\begin{equation} \label{proof_nec_eqn3}
		\int_r^{R_N} \frac{m(r)\rho(r)}{r^2} e^{3\lambda(r)} dr 
		\ \leq \
		\Phi(r) e^{\lambda(r)} -  \sum_{i=k}^{N-1} (\Phi(R_{i}^-)-\Phi(R_{i}^+))  e^{\lambda(R_{i})}  
	\end{equation}
Now, observing that $\Phi(R_i^-)\geq\Phi(R_i^+)$, which follows from  $\rho(R_i^-)\geq\rho(R_i^+)$ in combination with the RH-condition, estimate \eqref{proof_nec_eqn3} reduces to 
\begin{equation}   \label{proof_nec_eqn4}
		\int_r^{R_N} \frac{m(r)\rho(r)}{r^2} e^{3\lambda(r)} dr  
		\ \leq \ \Phi(r) e^{\lambda(r)}.
\end{equation}  
Using finally that boundedness of $p(r)=\frac{\rho(r)\Phi(r)}{\rho(r)-\Phi(r)}$ implies $\rho(r) > \Phi(r)$ for any $r \in [0,R)$, the sought-after inequality \eqref{necessary2} readily follows.  
\hfill $\Box$ \\

Note that \eqref{proof_nec_eqn3} is a necessary condition stricter than condition \eqref{necessary2}, but less convenient to check. 
Moreover, note that \eqref{phi_sec6} can be written equivalently as $\frac{1}{\Phi(r)}= \frac{1}{p(r)} + \frac{1}{\rho(r)}$, from which in combination with \eqref{proof_nec_eqn4} one obtains the following lower bound on the pressure, $r \in [0,R]$,
\begin{equation}
		p(r)\geq e^{-\lambda(r)}\int_r^{R} \frac{m(r)\rho(r)}{r^2} e^{3\lambda(r)} dr
= \left(1-\frac{2m}{r}\right)^{\frac{1}{2}}  \int^{R}_r \frac{m\rho}{r^2}  \left(1-\frac{2m}{r}\right)^{-\frac{3}{2}} dr.
	\end{equation}
Note finally that, if the  dominant energy condition $p(r)\leq\rho(r)$ holds, then \eqref{proof_nec_eqn4} leads to the following stricter necessary condition for $r \in [0,R)$
\begin{equation}
e^{-\lambda(r)}\int_r^R \frac{m(r) \rho(r)  }{r^2} e^{3\lambda(r)} dr < \frac{\rho(r)}{2}.
\end{equation}

\section{Examples and Applications}\label{Sec_app_examples}

We now investigate several natural examples of density profiles with the goal to better understand the constraint imposed by our sufficient condition \eqref{Thm1} for existence of finite pressure functions (and hence of stellar structures). These examples further illustrate the difference between our necessary and our sufficient conditions, \eqref{necessary_cond_Thm} and \eqref{Thm1}, as well as their relation to the classical Buchdahl limit. 

\subsection*{Example 1:}
We first illustrate that our condition \eqref{Thm1} is a refinement of the classical Buchdahl condition. For this, consider density functions with the property that 
\begin{equation} 
	m(r)=Ar^\alpha
\end{equation} 
such that $1<\alpha\leq3$. When $\alpha=3$, the density $\rho(r)$ is a constant on $[0,R]$, and \eqref{Thm1} reduces to the classical Buchdahl condition, cf. the discussion below Theorem \ref{Thm_existence}.  Note, if $\alpha>3$, $\rho(r)$ is increasing on $[0,R]$, a case we consider unphysical and which we omit in our setting. Denoting $M=m(R)=AR^\alpha$, the integral in \eqref{Thm1} can be computed as  
\[
\int_{0}^{R}Ar^{\alpha-2}(1-2Ar^{\alpha-1})^{-\frac{3}{2}}dr 
=\frac{1}{\alpha-1}\Big(\big(1-\frac{2M}{R}\big)^{-\frac{1}{2}}-1\Big).
\]
Thus, setting for simplicity $\Delta=1$ and using a strict inequality, condition \eqref{Thm1} can be written as
\begin{equation} \label{refined_Buchdahl}
	\frac{2M}{R} <  1-\frac{1}{\alpha^2},
\end{equation}
which is a refinement of the classical Buchdahl condition $\frac{2M}{R}< \frac{8}{9}\approx0.88889$.

\subsection*{Example 2}
Assume the density profile takes the Tolman VII form (cf. \cite{Raghoonundun2015}), 
\begin{equation}\label{model2}
	\rho(r) = \rho_c \left(1- \frac{r^2}{R^2}\right)
\end{equation}
for $r\in[0,R]$, where $\rho_c >0$ is the stellar central density.\footnote{Similar analysis can be applied for the case that the density $\rho(r)$ dose not vanish at the stellar boundary, i.e., $\rho(R)\not=0$.}  Direct integration yields
\begin{equation} \label{example2_eqn1}
	m(r)=4\pi r^3 \rho_c \left(\frac{1}{3}-\frac{r^2}{5R^2}\right),
\end{equation}
with $r\in[0,R]$. Since $\max_{r\in[0,R]}\frac{2m(r)}{r}=\frac{10}{9} \pi  \rho_c R^2$, condition \eqref{Outside_Schwarzschild} on the Schwarzschild radius reduces to $0 < \pi  \rho_c R^2 <\frac{9}{10}$; a constraint we henceforth assume. Now the integral in \eqref{Thm1} becomes
\begin{equation}
	\int^{R}_0 \frac{m(r)}{r^2} \left(1-\frac{2m(r)}{r}\right)^{-\frac{3}{2}} dr
	=
	\frac{(9-8\pi \rho_c R^2)(1-\frac{16}{15}\pi \rho_c R^2)^{-\frac{1}{2}}-9+20\pi \rho_c R^2}{4(9-10 \pi  \rho_c R^2)} \ \equiv \ f_1(x),
\end{equation}
for $x\equiv \pi \rho_c R^2$ with $0<x<\frac{9}{10}$.
Since $f_1(x)$ is strictly increasing on $(0,\frac{9}{10})$, we find that our sufficient condition \eqref{Thm1} is equivalent to $x\in(0,\frac{73-\sqrt{145}}{96})$, which we write as 
\begin{equation}
	0<\rho_c R^2< \tfrac{73-\sqrt{145}}{96 \pi } \approx 0.20212 
\end{equation}
for our condition that the pressure $p(r)$ is finite. From this, in combination with \eqref{example2_eqn1}, we directly obtain
\[
	\frac{2M}{R}=\frac{16}{15}\pi \rho_c R^2  < \frac{73-\sqrt{145}}{90}\approx 0.67732,
\]
as a minimal upper bound on mass radius relations of stellar models of form \eqref{model2}.

Next, based on the model \eqref{model2}, we study the necessary condition \eqref{necessary2} at $r=0$, which is given by 
\begin{equation}\label{model-ne}
	\int_0^R \frac{m(r) \rho(r) }{r^2}  \left(1-\frac{2m}{r}\right)^{-\frac{3}{2}}  dr 
	< \rho_c.
\end{equation}
Integration shows that \eqref{model-ne} is equivalent to 
\begin{equation}
	f_2(x)\equiv \frac{40  x-33+\sqrt{225-240 x}}{72-80 x}
	+\sqrt{\frac{5}{128x}}
	\ln\left(\frac{\sqrt{6}\sqrt{15-16 x}+2 \sqrt{x}}{3 \sqrt{10}-10 \sqrt{x} }\right)-1 <0 ,
\end{equation}
where again $x=\pi\rho_c R^2$. Since $f_2(x)$ is strictly increasing on $(0,\frac{9}{10})$, we find by numerical computation that $f_2(x)<0$ provided $x < 0.78856$. From this we conclude with the boundaries
\begin{equation} \label{example2_eqn2}
	\rho_c R^2 < 0.25100
\hspace{1cm} \text{and} \hspace{1cm}
	\frac{2M}{R}< 0.84112.
\end{equation}
Equation \eqref{example2_eqn2} shows a difference of $~ 0.05$ between our necessary and our sufficient condition for the model \eqref{model2}, as well as a stricter condition on the mass radius relationship than the classical Buchdahl condition ($\frac{2M}{R}<8/9 \approx 0.88889$).

\subsection*{Example 3}   

For the last example assume the metric component takes the form
\begin{equation}
	e^{2\lambda(r)} = (1+k r^2)^n,
\end{equation}
where $k>0$ and $n\in\mathbb{N}$. Since $e^{-2\lambda(r)} = 1-\frac{2m(r)}{r}$, it follows that
\begin{equation}
	m(r)=\frac{r}{2}\left(1-\frac{1}{(1+kr^2)^n}\right)
\end{equation}
and by \eqref{TOV2}
\begin{equation}
	\rho(r)=\frac{1}{8\pi r^2} \left(1+\frac{(2n-1)k r^2-1}{(1+kr^2)^{n+1}}\right).
\end{equation} 
Note that $\rho(r)$ is decreasing for $r\in[0,\infty)$ and, using l'H\^opital's rule, we find the central density is given by $$\rho(0)=\frac{3kn}{8\pi}.$$ 
Now our sufficient condition \eqref{Thm1} gives rise to
\begin{align}\label{example3}
	\int_{0}^{R} \frac{m(r)}{r^2}\left(1-\frac{2m(r)}{r}\right)^{-\frac{3}{2}} dr
	& =  \int_{1}^{1+kR^2} \frac{\tau^n -1}{4(\tau-1)} \tau^{\frac{3}{2} n} d\tau \cr  
	& = \sum_{i=1}^{n} \frac{1}{2n+4i}\left(e^{(1+\frac{2i}{n})\lambda(R)}-1\right) \ < \ 1,
\end{align}
by a change of variables to $\tau = 1+kr^2$. For $n=1$, \eqref{example3} is equivalent to 
\[
	\frac{1}{6}\left(e^{3\lambda(R)}-1\right)<1,
\]
which in turn is equivalent to $2M/R<1-7^{-\frac{2}{3}}\approx0.72672$, as a condition for the pressure to be bounded by Theorem \ref{Thm_existence}.
For $n=2$, \eqref{example3} is equivalent to
\[
	\frac{1}{8} e^{2\lambda(R)} + \frac{1}{12} e^{3\lambda(R)} -\frac{5}{24}<1,
\]
which in turn is equivalent to $2M/R	< 0.75662$. 
When $n=3$,  \eqref{example3} is equivalent to 
\[
	\frac{1}{10} e^{5\lambda(R)/3} + \frac{1}{14} e^{7\lambda(R)/3} + \frac{1}{18} e^{3\lambda(R)} - \frac{143}{630} < 1,
\]
which is equivalent to $2M/R<0.76787$, as a condition for the pressure to be bounded by Theorem \ref{Thm_existence}. 
For large $n$, by considering the Riemann sum, we can show \eqref{example3} converges to
\begin{equation}
	\frac{1}{4}\left(\operatorname{Ei}(3\lambda(R))-\operatorname{Ei}(\lambda(R)) - \ln3 \right)<1,
\end{equation}
where $\mathrm{Ei}(x)=\int _{-\infty }^{x} \frac{e^{t}}{t} dt$ denotes the exponential integral. Numerically we obtain $2M/R<0.79225$.

\bibliographystyle{plain}
\bibliography{BIBTEX.bib}

\end{document}